\documentclass[12pt]{article}
\usepackage{multibib}

\usepackage{float}
\usepackage[export]{adjustbox}
\usepackage{graphicx}
\graphicspath{ {./images/} }
\usepackage{subcaption}
\usepackage[utf8]{inputenc}
\usepackage[export]{adjustbox}
\usepackage{wrapfig}
\usepackage{authblk}
\usepackage{pgfplots}
\pgfplotsset{width=10cm,compat=1.9}
\usepackage[top=27mm, bottom=27mm, left=22mm, right=22mm]{geometry}
\usepackage{amsthm,amsmath,amssymb}
\usepackage{mathrsfs}
\usepackage{multirow}
\usepackage{amsthm}
\usepackage{microtype}
\usepackage{hyperref}
\usepackage[capitalize]{cleveref}
\usepackage{xcolor}

\newtheorem{theorem}{Theorem}[]
\newtheorem{lemma}[theorem]{Lemma}

\newtheorem{corollary}[theorem]{Corollary}

\newtheorem{definition}[theorem]{Definition}

\newtheorem{claim}[theorem]{Claim}

\title{List-decoding and list-recovery of Reed-Solomon codes beyond the Johnson radius for any rate }
\date{}

\author[$*$]{\textbf{Eitan Goldberg}}
\author[$\dagger$]{\textbf{Chong Shangguan}}
\author[$*$]{\textbf{Itzhak Tamo}}
\affil[$*$]{\footnotesize Department of Electrical Engineering–Systems, Tel-Aviv University, Tel-Aviv 39040, Israel.}
\affil[$\dagger$]{Research Center for Mathematics and Interdisciplinary Sciences, Shandong University, Qingdao 266237, China.}
\affil[$*\dagger$]{\footnotesize Emails: eitang1@mail.tau.ac.il,  theoreming@163.com, zactamo@gmail.com}

\begin{document}
\maketitle

\begin{abstract}
Understanding the limits of list-decoding and list-recovery of Reed-Solomon (RS) codes is of prime interest in coding theory and has attracted a lot of attention in recent decades. However, the best possible parameters for these problems are still unknown, and in this paper, we take a step in this direction. We show the existence of RS codes that are list-decodable or list-recoverable beyond the Johnson radius for \emph{any} rate, with a polynomial field size in the block length. In particular, we show that for any $\epsilon\in (0,1)$ there exist RS codes that are list-decodable from radius $1-\epsilon$ and rate less than $\frac{\epsilon}{2-\epsilon}$, with constant list size. We deduce our results by extending and strengthening a recent result of Ferber, Kwan, and Sauermann on puncturing codes with large minimum distance and by utilizing the underlying code's linearity. 
\end{abstract}

\section{Introduction}

Reed–Solomon (RS) codes \cite{RS-codes} are a classical family of error-correcting codes that have been studied extensively since their introduction in the 1960s. They have found many applications in both theory and practice (see for example \cite{mceliece1981sharing,  wicker1999reed}). In this paper, we consider these codes in the context of list-decoding and list-recovery to understand their performance under these problems. We begin with the needed definitions. 
 
For a prime power $q$, let $\mathbb{F}_q$ be the finite field of order $q$. For two vectors $x,y\in\mathbb{F}_q^n$, the {\it Hamming distance} $d(x,y)$ is the number of coordinates in which they differ, i.e., $d(x,y)=|\{i: x[i] \neq y[i]\}|$, where for $1\le i\le n$, $x[i]$ is the $i$th coordinate of $x$. Given integers $1\leq k\leq n\le q$, an $[n,k]$  code $C$  over $\mathbb{F}_q$ is a $k$-dimensional subspace of $\mathbb{F}_q^n$, where $n$ and $k$ are called the length and the dimension of the code, respectively. 
The {\it rate} of $C$ is defined to be $R:=\log_q|C|/n$, and the (minimum) {\it distance} of $C$ is the minimum Hamming distance between any two distinct vectors (codewords) of it. 

An $[n,k]$-RS code over $\mathbb{F}_q$ with evaluation  vector $\alpha=(\alpha_1,\ldots,\alpha_n)\in \mathbb{F}_q^n$, where $\alpha_i\neq \alpha_j$ for all $i\neq j$, is the $k$-dimensional subspace  
$$\{(f(\alpha_1),\ldots,f(\alpha_n)):f\in\mathbb{F}_q[x],~\deg(f)<k\},$$
where $\mathbb{F}_q[x]$ is the space of polynomials with coefficients in $\mathbb{F}_q$. It is well-known that the minimum distance of an $[n,k]$-RS code equals $n-k+1$, achieving the Singleton bound \cite{singletong-bound} with equality. Therefore, it is an {\it MDS} (maximum distance separable) code. In other words, every RS code has the best possible rate-distance trade-off, and it is optimal for unique decoding when one is required to output a unique codeword, given a corrupted received vector. Furthermore, several efficient algorithms (see for example \cite{Peterson1960,Reed1978RSalgo,Welch-Berlekamp}) for the unique decoding of RS codes are known. Despite this state of affairs, a lot is still unknown in the relaxed version of the unique decoding  problem of RS codes, which are the \emph{list-decoding} problem and its generalization, the \emph{list-recovery} problem. 
 
\paragraph{List-decoding of RS codes.}
In the list-decoding problem, given a corrupted codeword, one is allowed to output a list of possible codewords, in contrast to unique decoding, where the output is one codeword, i.e., the list is of size one. More formally, given $0<r<1$ and $L\in \mathbb{N}$, a code $C\subseteq\mathbb{F}_q^n$ is said to be $(r,L)$ list-decodable if for any $y\in \mathbb{F}_q^n$, 
$$|\{c\in C:d(c,y)\leq rn\}|\leq L,$$ 
where $r$ and $L$ are called the (list-decoding) \emph{radius} and the \emph{list size}, respectively. The notion of list-decoding was introduced independently and Elias \cite{elias} and Wozencraft \cite{wozencraft} in the 1950s. The \emph{list-decoding capacity} theorem (see Theorem 7.4.1 of \cite{guruswami2019essential}) provides the limits of list-decodability for general codes. It states that any code of rate $R$ cannot be list-decoded from a radius larger than $1-R$ with a sub-exponential list size in the block length, whereas there are codes of rate $R$ that are list-decodable with a radius of $1-R-\epsilon$ for every $\epsilon>0$ and list  size $O(1/\epsilon)$.

Due to the importance of RS codes and their prevalence in theory and practice, it is natural to ask how well they perform in the list-decoding problem. Further, besides the mathematical appeal of this question, over the years, the list-decodability of RS codes has found applications in complexity and pseudorandomness \cite{cai1999hardness,STV01,SIVAKUMAR1999270,LP20}.
  
The first result regarding the list-decodability of RS codes that one needs to mention is the well-known Johnson radius  (see \cite{johnson1962new} and Theorem 7.3.3 of \cite{guruswami2019essential}), which indicates that any RS code of rate $R$ is $(1-\sqrt{R}, O(qn^2))$ list-decodable. An efficient algorithm complemented this result, 
given by Guruswami and Sudan \cite{Guru-sudan-algo}, 
list-decodes RS codes up to radius $1-\sqrt{R}$, matching the Johnson bound. These results imply that for any $\epsilon\in (0,1)$, there exist RS codes that are list-decodable up to radius $1-\epsilon$, with rate $\epsilon^2$ and polynomial list size. On the other hand, no RS code can be list-decoded from radius $1-\epsilon$ with a rate larger than $\epsilon$ and a sub-exponential list size. Understanding the exact behavior of RS codes under list-decoding and possibly closing the gap from $\epsilon^2$ to $\epsilon$ is of great interest. 

Several results have indicated that a complete answer to this question is delicate. Ben-Sasson, Kopparty, and Radhakrishnan \cite{Ben-Sasson} showed that for every $\alpha>0$ there exist RS codes of rate $\epsilon^{2-\alpha}$, evaluated at all the elements of the field (namely, full-length RS codes), that are \emph{not} list-decodable from radius $1-\epsilon$. In other words, full-length RS codes can not be list-decoded {\it well} beyond the Johnson radius. However, this negative result left the possibility that some shorter RS codes can still be list-decoded  beyond this radius. This was later shown to be true by Rudra and Wootters \cite{Rudra-Wootters}, who showed that there are $(1-\epsilon,O(\frac{1}{\epsilon}))$ list-decodable RS codes with rate $\epsilon/(\log(q)\log^{5}(1/\epsilon))$. Recently, this was improved by Guo, Li, Shangguan, Tamo, and Wootters \cite{Guo-Li-Shangguan}, who exhibited the existence of $(1-\epsilon, O(\frac{1}{\epsilon}))$  list-decodable RS with rate $\Omega(\epsilon/\log (1/\epsilon))$, matching the list-decoding capacity up to a logarithmic factor. In the other regime of constant list size, Shangguan and Tamo \cite{shangguan2019combinatorial} showed that over an exponentially large field size, there exist $(\frac{L}{L+1}(1-R),L)$ list-decodable RS codes of rate $R$ for $L=2,3$, where the relation of decoding radius, rate, and list size was also shown to be optimal.
 
The last result in this sequence of improvements is the recent result by Ferber, Kwan, and Sauermann \cite{ferber2020}, who were the first to remove the logarithmic factor in the result of \cite{Guo-Li-Shangguan}, and thereby attain the list-decoding capacity up to a constant factor.  More precisely, they showed the existence of $(1-\epsilon, O(1/\epsilon))$ list-decodable RS codes with rate  $\epsilon/15$. Although this was a major improvement, the rate of the code rendered to be upper bounded  by $1/15$. Considering the practicality of these codes, it is of interest to understand whether one can improve the rate, as no RS code used in practice has such a low rate. We should also mention that Ferber, Kwan, and Sauermann in \cite{ferber2020} did not attempt to optimize their parameters to possibly obtain codes with larger rates. A careful  analysis of their method and a specific choice of optimized  parameters in their result yields to an improved rate of  $\epsilon/2$, which implies an upper bound of $1/2$ on the rate of  these RS codes.
In this paper, we push the rate substantially beyond the $1/2$ barrier by improving their method. In fact, we obtain our result for the more general problem of \emph{list-recovery} of RS codes, introduced next. 
 
\paragraph{List-recovery of RS codes.}   

A code $C\subseteq \mathbb{F}_q^n$ is said to be $(r,\ell,L)$ list-recoverable for $r\in (0,1)$ and $\ell,L\in \mathbb{N}$, if for every $n$ subsets $S_1,\ldots,S_n\subseteq \mathbb{F}_q, |S_i|\leq \ell$, $$|\{c\in C:c_i\notin S_i~\text{for at most }~rn ~\text{coordinates}\}|\leq L.$$ 
Evidently, list-recovery is a generalization of the list-decoding problem, since an $(r,1,L)$ list-recoverable code is also $(r,L)$ list-decodable. However, much less is known about  list-recovery  in contrast to list-decoding. A natural generalization of the Johnson bound for list-recovery (see~\cite{GS01}) implies that  any RS code of rate $R$ is list-recoverable from radius $1-\sqrt{\ell R}$ with a polynomial list size $L$. Equivalently, an RS code of rate  $\epsilon^2/\ell$ is list-recoverable up to radius $1-\epsilon$ with input list size $\ell$ and polynomial list size. Also, the Guruswami-Sudan algorithm is, in fact, a list-recovery algorithm that enables an efficient list-recovery up to the Johnson radius. Like list-decoding, Guruswami and Rudra \cite{Guruswami-rudra-limits-list-decoding} showed that some RS codes are not list-recoverable beyond the Johnson radius. None of the works
\cite{Rudra-Wootters,shangguan2019combinatorial,ferber2020} mentioned above considered the more general problem of list-recovery. Yet, recently, Lund and Potukuchi \cite{LP20} proved a list-recovery result, which is analogous to the result of \cite{Rudra-Wootters}. Specifically, they showed the existence of RS codes that are list-recoverable beyond the Johnson radius.  
Similarly to \cite{Rudra-Wootters}, the rate in their result has a $\log(q)$ factor in the denominator, which renders the RS code to have a vanishing rate. In addition, their result applies only to  radius less than $1-1/\sqrt{2}$, and not to any positive radius, as we consider. The current state-of-the-art for list-recoverability is the new work by Guo, Li, Shangguan, Tamo, and Wootters \cite{Guo-Li-Shangguan}, who showed the existence of RS codes of rate $\Omega(\frac{\epsilon}{\sqrt{\ell}\log(1/\epsilon)})$ that are $(1-\epsilon,\ell,O(\ell/\epsilon))$ list-recoverable, which holds only over exponentially large finite fields. 

\subsection{Contribution}
In this paper, we  establish the existence of RS codes that are list-decodable and list-recoverable from a large radius, which in many cases is the largest known radius. Our technique heavily relies on the result of Ferber, Kwan, and Sauermann \cite{ferber2020}. We strengthen their result on list-decoding by optimizing the choice of parameters in their proof,  which yields an improved dependence between the rate and the list-decoding radius. However, this is not sufficient to obtain  our final result. Our main technical contribution, which allows us to improve their result further,  utilizes the linearity of the underlying code and shows that by an encoding argument, one can further improve the code's list-decoding radius. We, in fact, prove the result for the more general problem of list-recovery, as stated below.

\begin{theorem}[Informal]
\label{informaltheorem}
    For any $\epsilon>0$ and $\ell\cdot R<1$ there exist $(1-\frac{\ell+1}{R+1}R-\epsilon,\ell,O(\frac{\ell}{\epsilon}))$ list-recoverable RS codes of rate $R$ over a polynomial (in the length of the code) field size. Equivalently, there exist  $(1-\epsilon,\ell,L)$ list-recoverable RS codes with rate approaching    $\frac{\epsilon}{\ell+1 -\epsilon}$, and  $L$ that depends only on the gap of the  rate to   $\frac{\epsilon}{\ell+1 -\epsilon}$.
\end{theorem}
See Corollary \ref{corRSexist} for the exact statement.
Figure \ref{figure_list_recovery} shows the  list-recovery radius as a function of the rate for $\ell=2$,  as given by Theorem \ref{informaltheorem}. Notice that the Johnson radius $1-\sqrt{\ell R}$ is non-negative only for rates $R\leq 1/\ell$, and in this range, one can verify that $1-\sqrt{\ell R}\leq 1-\frac{\ell+1}{R+1}R$, hence Theorem \ref{informaltheorem} always outperforms the Johnson radius.  Also, it is interesting to note that 
Guruswami and Rudra \cite{Guruswami-rudra-limits-list-decoding} showed the existence of RS codes for any rate greater than $1/\ell$ that are \emph{not} list-recoverable. On the other hand, Theorem \ref{informaltheorem} is a somewhat complementary result to it, as we show that for any rate $R$ smaller than $1/\ell$, there are list-recoverable RS codes with radius $1-O(R)$.

Our second main result   follows by specializing the above result on list-recovery to list-decoding, i.e., setting $\ell=1$. The exact statement of the following result appears in Corollary \ref{cor_RS_exist_listdecoding}.
\begin{theorem}[Informal] \label{stamstam}
For any $\epsilon>0$ there exist  $(1-\frac{2}{R+1}R-\epsilon,O(\frac{1}{\epsilon}))$ list-decodable RS codes of rate $R$ over a polynomial (in the length of the code) field size.  Equivalently, there are $(1-\epsilon,L)$ list-decodable RS codes with  rate approaching $\frac{\epsilon}{2-\epsilon}$, and $L$  that  depends only on the gap of the rate to   $\frac{\epsilon}{2-\epsilon}$.
\end{theorem}

\begin{figure}[h!]
    \centering
    \label{fig:my_label}
     \begin{tikzpicture}
       \begin{axis}
      [xlabel= Rate $R$,ylabel= List-recovery radius,]
        \addplot[
         domain=0:1,
         samples=100,
         color=black,
         ]
         {1-x};
      \addlegendentry{$1-R$}
      \addplot[
      domain=0:0.5,
      samples=100,
      color=red,
      ]
      {1-(2*x)^0.5};
      \addlegendentry{$1-\sqrt{2R}$}
      \addplot[
      domain=0:0.5,
      samples=100,
      color=green,
      ]
      {1-(3*x)/(x+1)};
     \addlegendentry{$1-\frac{3}{R+1}R$, Cor. \ref{corRSexist}}
     \end{axis}
    \end{tikzpicture}
    \caption{Known upper and lower bounds on the list-recovery of RS codes for $\ell=2$. The red curve is the Johnson radius,  the green curve is our result  (Corollary \ref{corRSexist}), and the black curve is the list-recovery capacity (upper bound)}
       \label{figure_list_recovery}
\end{figure}
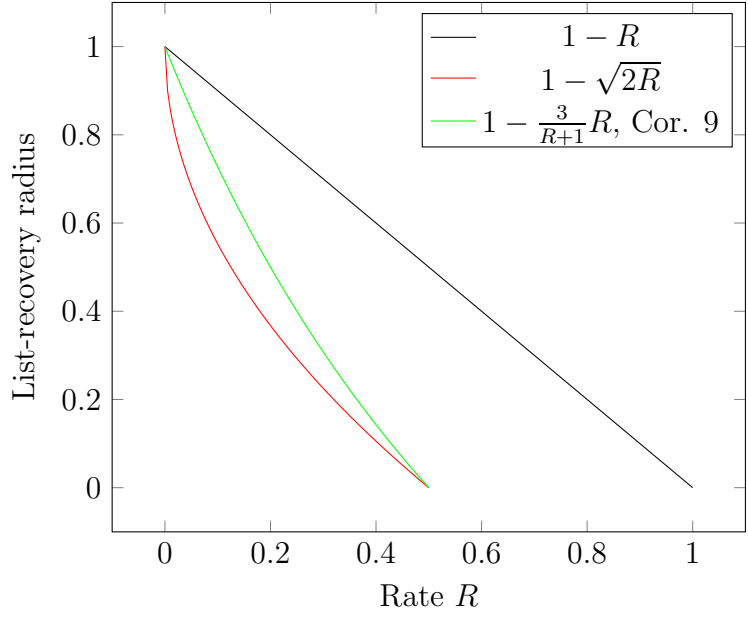

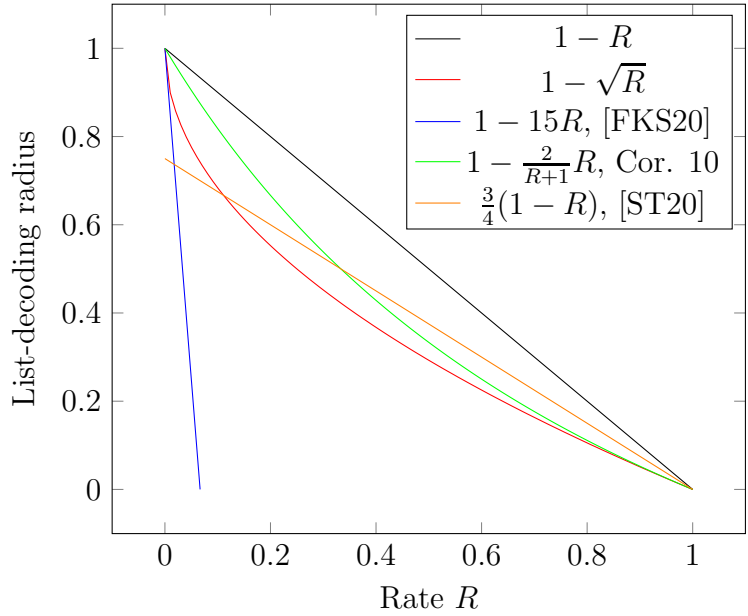
\begin{figure}[h!]
    \centering
    \label{fig:my_label}
    \begin{tikzpicture}
     \begin{axis}
      [xlabel= Rate $R$,ylabel= List-decoding radius,]
        \addplot[
         domain=0:1,
         samples=100,
         color=black,
         ]
         {1-x};
      \addlegendentry{$1-R$}
      \addplot[
      domain=0:1,
      samples=100,
      color=red,
      ]
      {1-x^0.5};
      \addlegendentry{$1-\sqrt{R}$}
      \addplot[
      domain=0:1/15,
      samples=100,
      color=blue,
      ]
      {1-15*x};
      \addlegendentry{$1-15R$, \cite{ferber2020}}
      \addplot[
      domain=0:1,
      samples=100,
      color=green,
      ]
      {1-(2*x)/(x+1)};
      \addlegendentry{$1-\frac{2}{R+1}R$, Cor. \ref{cor_RS_exist_listdecoding}}
       \addplot[
         domain=0:1,
         samples=100,
         color=orange,
         ]
         {(3/4)*(1-x)};
      \addlegendentry{$\frac{3}{4}(1-R)$, \cite{shangguan2019combinatorial} }
     \end{axis}
    \end{tikzpicture}
        \caption{ Known upper and lower bounds on  the list-decoding of RS codes. The black curve is the list-decoding capacity. The  other plots are all lower bounds. The  red curve is   the Johnson bound; the blue curve is the result of Ferber, Kwan, and Sauermann  \cite{ferber2020}, the orange curve is the result of Shangguan and Tamo  \cite{shangguan2019combinatorial}, and  the green curve is our new result in Corollary \ref{cor_RS_exist_listdecoding}.}
   \label{fig:list_decoding}
\end{figure}

As already mentioned, the main result of Ferber, Kwan, and Sauermann \cite{ferber2020} showed that there exist RS codes with rate $\epsilon/15$ that can be list-decoded from radius $1-\epsilon$ and list size $3/\epsilon$. Hence, the rate of $\frac{\epsilon}{2-\epsilon}$ achieved in Theorem \ref{stamstam} is a major improvement of the rate compared to \cite{ferber2020}. 
Furthermore, to the best of our knowledge, this is the first existence result of RS codes that are list-decodable beyond Johnson radius for \emph{any} rate with polynomial list size.
Other results, such as Guo, Li, Shangguan, Tamo, and Wootters \cite{Guo-Li-Shangguan}, and Shangguan and Tamo \cite{shangguan2019combinatorial} exceed it for a wide range of rate values, but not for all values, and their result also requires an exponential field size. 
 Figure \ref{fig:list_decoding} plots   the list-decoding radius as a function of the rate (in green),  as given by Theorem \ref{stamstam}. The other plots are the list-decoding capacity (upper bound) and the other  known lower bounds. It can be seen that, indeed, the curve  derived from  Theorem \ref{stamstam}  is the only one that exceeds the Johnson radius for any rate. 
We summarize the known results and the results provided in this paper in  Table \ref{tab:litreview}.
\renewcommand{\arraystretch}{1.4}
\begin{table}[H]
\centering
\begin{tabular}{|p{3.8cm}||c|c|c|c|}
\hline
 & Radius $r$ & List size $L$ & Rate $R$ & Field size $q$ \\
\hline\hline
\textbf{List-Decoding:} &&&&\\
\hline
Capacity 
& $1 - \epsilon$ & - & $\leq  \epsilon$ & - \\
\hline
Johnson bound & $1 - \epsilon$ & $\text{poly}(n)$ & $C \epsilon^2$ & $q \geq n$ \\
\hline
\cite{Rudra-Wootters} & $1 - \epsilon$ & $\frac{C}{\epsilon}$ & $\frac{C\epsilon}{ \log^5(1/\epsilon) \log(q) }$ & $q \geq Cn\log^C(n/\epsilon)/\epsilon$ \\
\hline 
\cite{shangguan2019combinatorial} & $\frac{L}{L+1}(1-R)$ & $L=2,3$ & $R$ & $q = 2^{Cn}$ \\
\hline
\cite{Guo-Li-Shangguan} & $1 - \epsilon$ & $\frac{C}{\epsilon}$ & $\frac{C\epsilon}{ \log(1/\epsilon) }$ & $q = \left(\frac{1}{\epsilon}\right)^{Cn}$ \\
\hline
\cite{ferber2020} & $1-\epsilon$& $\lceil\frac{3}{\epsilon}\rceil $& $\frac{\epsilon}{15}$ & $q\ge n^{1.25}$\\
\hline
Our work Cor. \ref{cor_RS_exist_listdecoding} & $1-\epsilon$ & $\lfloor\frac{2}{\zeta}\rfloor,~\forall \zeta>0$& $\frac{\epsilon-\zeta}{2-\epsilon+\zeta}$ & poly($n$)\\
\hline
\hline 
\textbf{List-Recovery:} &&&&\\
\hline
Capacity  & $1 - \epsilon$ & - & $\leq  \epsilon$ & -  \\
\hline
Johnson bound & $1 - \epsilon$ & $\text{poly}(n)$ & $\frac{C\epsilon^2}{\ell}$ & $q \geq n$ \\
\hline
\cite{LP20} & $r\leq 1 - 1/\sqrt{2}$& $C\ell$ & $\frac{C}{\sqrt{\ell} \cdot \log q}$ & $q \geq C n \sqrt{\ell} \cdot \log n$ \\
\hline 
\cite{Guo-Li-Shangguan} & $1 - \epsilon$ & $\frac{C\ell}{\epsilon}$ & $\frac{C\epsilon}{\sqrt{\ell} \cdot \log(1/\epsilon)}$ & $q = \left(\frac{\ell}{\epsilon} \right)^{Cn}$ \\
\hline
Our work  Cor. \ref{corRSexist} & $1-\epsilon$ & $\lfloor\frac{2\ell}{\zeta}\rfloor,\forall \zeta>0$& $\frac{\epsilon-\zeta}{l+1-\epsilon+\zeta}$ & poly($n$)\\

\hline
\end{tabular}
\caption{Known results on list-decoding and list-recovery of RS codes.  The ``Capacity'' results are upper bounds on the rate-radius trade-off, $C$ is   an absolute constant,   and  $n$ is sufficiently large  relative to $1/\epsilon$ and $\ell$.
}
\label{tab:litreview}
\end{table}

 %

\newpage
\section{Definitions and background}
\label{notation}
We will need the following notations. For a positive integer $m$ let $[m]=\{1,\ldots,m\}$; for a vector $v\in \mathbb{F}_q^m$ we will also use the notation $v[i]$ for the $i$th coordinate of $v$. For a  vector $a=(a_1,\dots,a_n)$, with $a_i\in [m]$ and $a_i\neq a_j$ for all $i\neq j$, let $$v_a=(v[a_1],\ldots,v[a_n]);$$ for a code   $C\subseteq \mathbb{F}_q^m$,  let $$C_a=\{c_a: c\in C\}.$$    
Recall that an $[n,k]$ code is called an MDS code if its minimum distance $d$ attains the Singleton bound, i.e., $d=n-k+1.$

\begin{definition}
Let $\ell,L,h\in \mathbb{N}, c\in \mathbb{R}$. We say that a vector $a=(a_1,\dots,a_n)$ is a bad puncturing with the certificate $(\mathcal{I},S)$ if 

\begin{enumerate}
\item $\mathcal{I}$ is a family of $L+1$ subsets $I_k\subseteq [\ell]\times [n],k=1,\ldots,L+1$ such that
\begin{equation}
\label{first}    
\sum_{k=1}^{L+1} |I_k|-|\bigcup_{k=1}^{L+1} I_k|>chL.
\end{equation}
    \item $S\in \mathbb{F}_q^{\ell\times n}$ is a $q$-ary $\ell\times n$ matrix with distinct entries in each column.
    \item There exist $L+1$ codewords $\gamma_1,\ldots,\gamma_{L+1}\in C$  such that for  $k=1,\dots, L+1$ 
\begin{equation}
\label{third}
I_k=\{(i,j)\in [\ell]\times [n]:\gamma_k[a_j]=S_{ij}\},
\end{equation}
where $S_{ij}$ is $(i,j)$th entry of the matrix $S$. 
\end{enumerate}
\end{definition}

We will make use of the well-known Chernoff bound.

\begin{lemma}[Chernoff bound, see for example Theorem A.1.4 of \cite{alon2016probabilistic}] \label{chrenoff}
Let $X_1,\ldots,X_s$ be independent Bernoulli random variables with ${\rm Pr}(X_i)=p$ for all $i$, then for all $\epsilon\in [0,1]$ $${\rm Pr}\left[\frac{1}{s}\sum\limits_{i=1}^{s}X_i > p+\epsilon\right ]<e^{-2\epsilon^2 s},$$
and $${\rm Pr}\left[\frac{1}{s}\sum\limits_{i=1}^{s}X_i < p-\epsilon\right]<e^{-2\epsilon^2 s}.$$
\end{lemma}

\section{Results and proofs}
This section contains precise statements of our results and their proofs. The results will follow from the following theorem that shows that a random puncturing of a code with a large minimum distance is not bad with high probability, where we say that a puncturing is bad if it is bad for some certificate.


\begin{theorem}
\label{thmmainlistrecovery}
Let $C\subseteq \mathbb{F}_q^m$ be a linear code with minimum distance $m-h$ and rate $\frac{Rn}{m}$ for $n\in [m]$. Let $\ell,L\in \mathbb{N}$ and $c,c'\in\mathbb{R}$ be constants that satisfy $\frac{\ell+1-c}{c}<R< \frac{\ell}{c}$ and $1<c'<\frac{c-1}{\ell-Rc}$. Assume further that $h \leq q^{-\frac{1}{c'}}m$,
then there are at most $2^{(L+1)\ell n}q^{-\alpha h}m^n$ bad puncturings of $C$, where  $\alpha=\alpha(\ell,c,c',R)>0$. 
\end{theorem}


Note that by choice of the parameters we have  $1<\frac{c-1}{\ell-Rc}$,  therefore  one can pick parameter $c'$ in the required range. We will assume that the parameters $h,n,q,m$ are all large enough compared to the fixed constants $\ell,L,c,c',R$. 

\begin{proof} The result will follow by showing that for each family of sets $\mathcal{I}$ that satisfies \eqref{first} there are at most $q^{-\alpha h}m^n$ bad puncturings for $C$ with a certificate $(\mathcal{I},S)$ for some matrix $S\in \mathbb{F}_q^{\ell\times n}$. Then, since the number of such $\mathcal{I}$'s is at most $2^{(L+1)\ell n}$ the result will follow. 

Fix an $\mathcal{I}$ that satisfies $\eqref{first}$ and let $a=(a_1,\ldots,a_n)\in [m]^n$ be a bad puncturing of $C$ with a certificate $(\mathcal{I},S)$. Note, first that for any $I\in \mathcal{I}$ there are no two distinct elements  $(i_1,j),(i_2,j)\in I,$ with $i_1\neq i_2$, since there is no codeword $\gamma \in C$ such that $\gamma[i]=S_{i_1j}$ and $\gamma[i]=S_{i_2j}$, as $S$ has distinct column entries. Therefore, we may assume that any two distinct elements of any  $I\in \mathcal{I}$ do not agree on their second coordinate.

We proceed by induction on  $L$. For $L=0$, there is no $\mathcal{I}$ that satisfies \eqref{first}, and the result holds trivially.  
Next, assume that $L>0$ and that the claim holds for $L-1$. 
If $\mathcal{I}$ contains a set, which we assume without loss of generality to be $I_{L+1}$, that satisfies  $|I_{L+1}\cap\bigcup_{k\in[L]} I_k|<ch$, then  
\begin{align*}
\sum_{k=1}^{L}|I_k| -|\bigcup_{k\in[L]} I_k|=\sum_{k=1}^{L+1}|I_k| -|\bigcup_{k\in[L+1]} I_k|-|I_{L+1}\cap\bigcup_{k\in[L]} I_k|> chL-ch=ch(L-1).
\end{align*}
Then, the vector $a$ is also a bad puncturing for $C$ with the certificate $(\mathcal{I}\backslash I_{L+1},S)$, and therefore by the induction hypothesis
there are at most $q^{-\alpha h}m^n$ bad puncturings $a$ of $C$. 


Next, we assume that for any $k\in[L+1]$, $|I_{k}\cap\bigcup_{k\ne k'} I_{k'}|\ge ch$. Let $\pi:[\ell]\times [n]\to [n]$ be the projection on the second coordinate, i.e., $\pi(i,j)=j$. Let $M\subseteq [n]$, $M:=\pi(\bigcup_{k\ne k'\in[L+1]} I_k\cap I_{k'})$, and notice that for any $I\in\mathcal{I}$,
\begin{equation}
\label{second}
    |\pi(I)\cap M|\ge ch. 
\end{equation}
We will need the following claim that shows  that $M$ contains a relatively small subset that has a large intersection with every $\pi(I),~I\in\mathcal{I}$. The  proof of the claim is very similar to the proof of Claim 4 in \cite{ferber2020}.
 
\begin{claim}
\label{claimZ}
There is a set $Z\subseteq  M$ and $\lambda_1>1/c$ such that $|Z|\leq \frac{|M|}{\lambda_1c}$ and $|Z\cap I_k|> h$ for all $k\in [L+1]$.
\end{claim}
\begin{proof}
Since $c'<\frac{c-1}{\ell-Rc}$ and $Rc<\ell$, then $\frac{\ell c'+1}{c(Rc'+1)}<1$. Hence, there exist $\lambda_1, \lambda_2$ such that 
   \begin{equation}
       \label{eq1.1}
       \frac{1}{c}<\frac{\ell c'+1}{c(Rc'+1)}<\lambda_1<\lambda_2<1,
   \end{equation}
where the first inequality follows as $R<\ell$.
Moreover, $\frac{1}{\lambda_2c}<1$.

Let $Z\subseteq M $ be a random subset formed by choosing each element of $M$ independently with probability  $\frac{1}{\lambda_2c}$. 
By Lemma \ref{chrenoff}, 
\begin{equation*}
{\rm Pr}\left[|Z|>\frac{|M|}{\lambda_1c}\right]\leq  e^{-2(\frac{1}{\lambda_1c}-\frac{1}{\lambda_2c})^2|M|}
\leq e^{-2(\frac{1}{\lambda_1}-\frac{1}{\lambda_2})^2\frac{h}{c}},\\
\end{equation*}
where the last inequality holds as by \eqref{second} we have $|M|\ge ch$.  Furthermore, it also follows by Lemma \ref{chrenoff} and \eqref{second} that
for all $k\in[L+1]$, 
\begin{equation*}
{\rm Pr}\left[|Z\cap \pi(I_k)|\leq h\right]\leq  {\rm Pr}\left[|Z\cap \pi(I_k)|\leq \frac{|M\cap \pi(I_k)|}{c}\right]
\leq e^{-2(\frac{1}{\lambda_2c}-\frac{1}{c})^2|M\cap \pi(I_k)|}\leq e^{-2(\frac{1}{\lambda_2}-1)^2\frac{h}{c}}.
\end{equation*}
 By  the union bound,  the probability that for all $k\in [L+1]$, $|Z\cap \pi(I_k)|> h$ and that $|Z|\leq \frac{|M|}{\lambda_1c}$ is at least $1-e^{-2(\frac{1}{\lambda_1}-\frac{1}{\lambda_2})^2\frac{h}{c}}-(L+1)e^{-2(\frac{1}{\lambda_2}-1)^2\frac{h}{c}}$, which is strictly positive
for large enough $h$. Hence, with a positive probability there exists a set $Z$ with the claimed properties.
\end{proof}
To each set $M$ as above, we associate a fixed subset $Z\subseteq M$, as given by  Claim \ref{claimZ}.

The next claim utilizes the additive structure of the linear code to show that any bad puncturing has many certificates. 
\begin{claim}
\label{claimZZ}
If  $a$ is a bad puncturing of  $C$ with the certificate $(\mathcal{I},S)$, then it is also bad with the certificate $(\mathcal{I},S+v_a)$ for any $v\in C$, where by abuse of notation $S+v_a$ is the matrix obtained by adding to each row of $S$ the vector $v_a$
\end{claim}

\begin{proof}
Let $\gamma_1,\ldots,\gamma_{L+1}\in C$ be codewords such that  for any $k=1,\ldots, L+1$, 
$$I_k=\{(i,j)\in [\ell]\times[n]:\gamma_k[a_j]=S_{ij}\}.$$ 
Fix any $k=1,\ldots,L+1$ and $(i,j)\in I_k$, then the codeword $\gamma_k':=\gamma_k+v\in C$ satisfies

$$(i,j)\in I_k\Longleftrightarrow \gamma_k[a_j]=S_{ij} 
\Longleftrightarrow \gamma_k[a_j]+v[a_j]=S_{ij}+v[a_j]
\Longleftrightarrow \gamma_k'[a_j]= (S+v_a)_{ij}.$$
Equivalently, 
$$I_k=\{(i,j)\in [\ell]\times[n]:\gamma_k'[a_j]=(S
+v_a)_{ij}\},$$
and the result follows. 
\end{proof}


Next, given a family of sets $\mathcal{I}$ that satisfies \eqref{first}, we give an encoding argument to all the bad puncturings with a certificate $(\mathcal{I},S)$, for some matrix $S\in \mathbb{F}_q^{\ell\times n}$.
Note that given $\mathcal{I}$, it is possible to find the set $M$, and then also its  associated fixed  subset $Z\subseteq M$ from  Claim \ref{claimZ}. 
Furthermore, given a certificate $(\mathcal{I},S)$ of a bad puncturing, one can determine the $L+1$ codewords $\gamma_k$ that satisfy \eqref{third}, due to the minimum distance of the code. Since one can determine the values $\gamma_k$ attained at more than $h$ coordinates, which in turn uniquely determines the codeword. We proceed to the encoding.

\paragraph{}
Fix a family of sets $\mathcal{I}$ that satisfies \eqref{first}. For a bad puncturing $a$ with a certificate $(\mathcal{I},S)$ do the following.  

\vspace{0.1cm}
{\bf Encoding:}
\begin{enumerate}
        \item  Encode  the values of the coordinates $a_i$ for $i\in ([n]\backslash M)\cup Z$. Since each $a_i\in [m]$, the encoding has   at most  $m^{n-|M|+|Z|}$ possibilities.
    
    \item Let $S_Z$ and $a_Z$ be the restriction of $S$ and $a$ to the columns and coordinates with indices in $Z$, respectively. 
   Let $Mat(C_{a_z})$ be  the space of all $\ell\times |Z|$ matrices over $\mathbb{F}_q$  whose rows are  $\ell$ identical codewords of  the punctured code $C_{a_Z}$. Since $|Z|>h$, the dimension of $C_{a_Z}$ and also of $Mat(C_{a_z})$ is  $Rn$. 
  Encode the coset of $Mat(C_{a_z})$ in the space $\mathbb{F}_q^{\ell \times |Z|}$  that contains the matrix $S_Z$. This encoding has $q^{\ell|Z|-Rn}$ possibilities.

    \item Encode the coordinates $a_i,i\in M\backslash Z$.
    By the minimum distance of the code $C$,  each  $a_i,i\in M\backslash Z$ has at most $h$ options, since for each $i\in M$ there are at least two distinct $\gamma_k$'s which agree on the coordinate $a_i$. Therefore, the encoding has at most   $h^{|M|-|Z|}$ possibilities. 
\end{enumerate}
To conclude, given  $\mathcal{I}$ the  encoding  is a mapping from the set of bad puncturings with a certificate $(\mathcal{I},S)$ for some matrix $S\in \mathbb{F}_q^{\ell\times n}$,   to the set $[m]^{n-|M|+|Z|}\times \mathbb{F}_q^{\ell\times |Z|}/Mat(C_{a_z}) \times[h]^{|M|-|Z|}$.
Thus, the total number of possible encodings is at most  
$$m^{n-|M|+|Z|}q^{\ell|Z|-Rn}h^{|M|-|Z|}.$$
Next, we show that the encoding is reversible, i.e., it is an injective mapping, and given the encoding of $a$, one can recover $a$.

\vspace{0.2cm}
{\bf Decoding:}
Recall that $\mathcal{I}$ is given, and therefore the set $M$ and its associated subset $Z$ are also known. Then, from step (1) of the encoding we can recover the restriction of $a$ to its coordinates in $([n]\backslash M)\cup Z$, i.e., we know the values of $a_i,~i\in ([n]\backslash M)\cup Z$. Hence, it remains to recover the subvector $a_{M\backslash Z}$. 
By step (2), let $(S+v_a)_Z$ be an arbitrary matrix in the in the coset that contains $S_Z$, where $v\in C$ is some codeword, and $S+v_a$ is the matrix formed by adding the codeword $v_a$ to all the rows of $S$.  By Claim \ref{claimZZ} the vector $a$ is also bad with the certificate $(\mathcal{I},S+v_a)$ and the $L+1$ codewords $\gamma_k':=\gamma_k+v.$ We claim that the $L+1$ codewords $\gamma_k'$ can be deduced from the encoding. Indeed, for any $j\in \pi(I_k)\cap Z$, it holds that $\gamma_k'[a_j]= (S+v_a)_{ij}$, where $i\in[\ell]$ is the unique element such that $(i,j)\in I_k$. Hence, one can recover $\gamma_k'$ since we know its value on $|\pi(I_k)\cap Z|>h$ coordinates. Lastly, by the knowledge of the $\gamma_k'$ for $k=1,\ldots,L+1$ and step (3), one can recover the remaining coordinates $a_i,~i\in M\backslash Z$. Note that the codewords $\gamma_i,\gamma_j$ agree on a coordinate if and only if the codewords $\gamma_i',\gamma_j'$ also agree on this coordinate, and this concludes the decoding. 
 
Since the encoding is an injective mapping, given a family of sets $\mathcal{I}$ that satisfies \eqref{first}, the number of bad puncturings $a$ with respect to $\mathcal{I}$ is at most the size of the image of the mapping, which is at most  

\begin{align}
 m^{n-|M|+|Z|}q^{l|Z|-Rn}h^{|M|-|Z|}
   &=(\frac{h}{m})^{|M|-|Z|}q^{\ell|Z|-Rn}\cdot m^n\nonumber\\
   &\leq q^{-\frac{1}{c'}(|M|-|Z|)}q^{\ell|Z|-R|M|}\cdot m^n\nonumber\\
   &\leq q^{\frac{\ell c'+1-\lambda_1c(1+Rc')}{\lambda_1cc'}|M|}\cdot m^n\nonumber\\
   &\leq q^{\frac{\ell c'+1-\lambda_1c(1+Rc')}{\lambda_1cc'}ch}\cdot m^n\nonumber\\
    &=q^{\frac{\ell c'+1-\lambda_1c(1+Rc')}{\lambda_1c'}h}\cdot m^n\nonumber,
\end{align} 
where the first inequality holds as by assumption $h\leq q^{-\frac{1}{c'}}m$ and $|M|\le n$, the second inequality holds as by Claim \ref{claimZ} $|Z|\le\frac{|M|}{\lambda_1 c}$, and the third inequality holds as by \eqref{eq1.1} $\ell c'+1-\lambda_1c(1+Rc')<0$ and by \eqref{second} $|M|\ge ch$. It follows that the number of bad puncturings is at most $q^{-\alpha h}m^n$, where $\alpha:=-\frac{\ell c'+1-\lambda_1c(1+Rc')}{\lambda_1c'}$. This concludes the induction step, and the result follows.
 \end{proof}
Next, we move to prove the paper's main result on the list-recoverability of linear codes. The result will follow by invoking Theorem \ref{thmmainlistrecovery}. We note that the theorem holds for any linear MDS codes; however, we state it specifically for RS codes.

\begin{theorem}\label{thmlistrecovery}
Let $C\subseteq \mathbb{F}_q^q$ be the full-length RS code of dimension $Rn$ with $n\in [q], R>0$. Let $\ell,L\in \mathbb{N}$ be positive integers and constants $c,c'$  that  satisfy $\frac{\ell+1-c}{c}<R< \frac{\ell}{c}$, $1<c'<\frac{c-1}{\ell-Rc}$, and $Rn\leq q^{1-\frac{1}{c'}}+1$ 
Then, for large enough  $n$ and $q$,  a random puncturing of $C$ to a code of length $n$ is a $(1-\frac{1}{L+1}(\ell+LcR),l,L)$ list-recoverable code with probability at least $1-q^{-\frac{\alpha}{4}Rn}$, for some positive $\alpha :=\alpha(\ell,c,c',R)$.
\end{theorem}

\begin{proof}
The result will follow by providing  an upper bound on the number of distinct puncturings of $C$ for which the resulting code in not  $(1-\frac{1}{L+1}(\ell+LcR),\ell,L)$ list-recoverable. Assume that for a vector  $a=(a_1,\ldots,a_n)\in \mathbb{F}_q^n$ with distinct entries the punctured code $C_a$ is not $(1-\frac{1}{L+1}(\ell+LcR),\ell,L)$ list-recoverable, then by definition  there exist  $n$ lists $S_j\subseteq \mathbb{F}_q$ of size $\ell$ each, and $L+1$ distinct code words  $\gamma_1,\ldots,\gamma_{L+1}\in C$, such that for any $k=1,\ldots,L+1$ the number of indices $j\in [n]$ for which $\gamma_k[a_j]\in S_j$ is at least $\frac{\ell n}{L+1}+\frac{LcRn}{L+1}$. 
Let $S$ be an $\ell\times n$ matrix whose $j$th column is the elements of the list $S_j$ ordered arbitrarily. 
Define for $k=1,...,L+1$ the set  $I_k\subseteq [\ell]\times [n]$ to be $$I_k:=\{(i,j):\gamma_k[a_j]=S_{ij}\},$$
and note that   $|I_k|\ge \frac{\ell n}{L+1}+\frac{LcRn}{L+1}$, hence 
$$\sum\limits_{k=1}^{L+1}|I_k|-|\bigcup\limits_{k=1}^{L+1}I_k|\ge \ell n+cRnL-\ell n=cRnL>c(Rn-1)L.$$
This implies that $a$ is a bad puncturing  with the  certificate $(\{I_1,...,I_{L+1}\},S)$. 

Since  $C$ is an MDS code,  its minimum  distance is $q-Rn+1$, then 
by  Theorem \ref{thmmainlistrecovery} with $m=q,h=Rn-1$, the number of such vectors $a$ is at most  \begin{equation}
\label{zzzz}
2^{(L+1)\ell n}q^{-\alpha h}m^n=q^{\ell(L+1)\log_q(2)n-\alpha(Rn-1)}q^n.
\end{equation}
For large enough $q$ (compared to $\ell,c,c',R,L$, which are viewed as constants)  
 $$\ell(L+1)\log_q(2)<\frac{\alpha}{3}R.$$
  Hence since $Rn>6$, then  \eqref{zzzz} is at most   $$q^{\frac{\alpha}{3}Rn-\alpha(Rn-1)}q^n\leq q^{-\frac{\alpha}{2}Rn}q^n.$$ 

Next, it is left to lower bound the total number of $n$-tuples with distinct entries, similar to  \cite{ferber2020}.
Since  $Rn\leq q^{1-\frac{1}{c'}}+1$, then for  large enough  $q$ compared to $n$
\begin{equation}
\label{x}
\frac{n}{q}\leq \min\{\frac{1}{2},\frac{\alpha R}{8}\}.
\end{equation}
Then, the number of $n$-tuples with distinct entries  is
\begin{align}
    q(q-1)\cdots(q-(n-1))&\ge (1-\frac{n}{q})^nq^n\nonumber\\
    & \ge 2^{-2\frac{n^2}{q}}q^n,\label{xx}\\
    & \ge 2^{-\frac{\alpha}{4}Rn}q^n \label{xxx}\\
    & \ge q^{-\frac{\alpha}{4}Rn}q^n \nonumber.
\end{align}
where \eqref{xx}  follows since  $1-x\ge 2^{-2x}$ for $x\in (0,\frac{1}{2})$ and \eqref{x}, and \eqref{xxx} follows by \eqref{x}.
Hence,  the probability of  a random puncturing not to be  $(1-\frac{1}{L+1}(\ell+LcR),\ell,L)$ list-recoverable  is at most $$\frac{q^{-\frac{\alpha}{2}Rn}q^n}{q^{-\frac{\alpha}{4}Rn}q^n}=q^{-\frac{\alpha}{4}Rn},$$
and the result follows.
\end{proof}

The next corollary  follows from  Theorem \ref{thmlistrecovery}, for  a large list size $L$ and the best  possible $c$ for a given rate $R$.
\begin{corollary}
\label{corRSexist}
For $\ell\ge 1, 0<R<\frac{1}{\ell},\epsilon>0$,    $n>n_0(\ell,R,\epsilon)$,  and field size $q\ge n^{\frac{c'}{c'-1}}$, there exist a $(1-\frac{\ell+1}{R+1}R-\epsilon,l,O(\frac{\ell}{\epsilon}))$ list-recoverable $[n,Rn]_q$ RS code, where $c'=c'(\ell,R,\epsilon)>1$.
\end{corollary}
\begin{proof}
Given $\epsilon,\ell$ and $R$ let $L$ be a positive integer  such that  $L+1\ge \frac{2\ell}{\epsilon}$, hence $L=O(\frac{\ell}{\epsilon})$.
Then, by Theorem \ref{thmlistrecovery} for constants $c,c'$ that satisfy $\frac{\ell+1}{R+1}<c<\frac{1}{R}, 1<c'<\frac{c-1}{\ell-Rc}$ and large enough $n$  and $q\ge n^{\frac{c'}{c'-1}}$,  there exists a  $(1-\frac{\ell+LcR}{L+1},\ell,L)$ list-recoverable $[n,Rn]_q$  RS code.
The result will follow by showing that $1-\frac{\ell+LcR}{L+1}\geq 1-\frac{\ell+1}{R+1}R-\epsilon$ for small enough  $c$.
Indeed, let $c>\frac{\ell+1}{R+1}$ be small enough such that $(c-\frac{\ell+1}{R+1})R\leq \frac{\epsilon}{2}$. Then, 
\begin{equation}
\label{t-1}    
1-cR\ge 1-\frac{\ell+1}{R+1}R-\frac{\epsilon}{2}.
\end{equation}
Moreover, by the choice of $L$,  
\begin{equation}
\label{t0}
1-cR-(1-\frac{\ell+LcR}{L+1})=\frac{\ell-cR}{L+1}\leq  \frac{\ell}{L+1}\leq  \frac{\epsilon}{2}.
\end{equation}
Hence, 
\begin{align}
    1-\frac{\ell+LcR}{L+1}&\geq  1 -cR -\frac{\epsilon}{2} \label{t1}\\
    &\ge 1-\frac{\ell+1}{R+1}R-\epsilon, \label{t2}
\end{align}
where \eqref{t1} and \eqref{t2} follow by \eqref{t0} and \eqref{t-1} respectively, and the result follows.
\end{proof}
The next corollary is analogous  to Corollary \ref{corRSexist}  for list-decoding, as it is a special case of list-recovery.  
\begin{corollary}
\label{cor_RS_exist_listdecoding}
Let  $0<R<1, \epsilon>0$, then for  large enough $n>n_0(R,\epsilon)$, such that $Rn\in \mathbb{N}$ and field size $q\ge n^{\frac{c'}{c'-1}}$, there  exists a $(1-\frac{2}{R+1}R-\epsilon,O(\frac{1}{\epsilon}))$  list-decodable $[n,Rn]$ RS code, where $c'>1$ is a function of $R$ and $\epsilon$.
\end{corollary}
\begin{proof}
Apply Corollary \ref{corRSexist} with $\ell=1$.
\end{proof}

\section*{Acknowledgements}
The  research  of Eitan Goldberg and Itzhak Tamo is partially supported by the European Research Council (ERC grant number 852953) and by the Israel Science Foundation (ISF grant number 1030/15).

The research of Chong Shangguan is supported by the Qilu Scholar Program of Shandong University and the National Key Research and Development Program of China under Grant No. 2020YFA0712100.


{\small\bibliographystyle{alpha}
\bibliography{new}}

\end{document}